\documentclass[submission,copyright,creativecommons]{eptcs}
\providecommand{\event}{TARK 2023} 
\usepackage{comment}
\usepackage{iftex}

\ifpdf
  \usepackage{underscore}         
  \usepackage[T1]{fontenc}        
\else
  \usepackage{breakurl}           
\fi
\usepackage{multicol}
\usepackage{breqn}
\usepackage{caption}
\usepackage{subcaption}
\usepackage{csquotes}
\usepackage{mathtools}
\usepackage{hyperref}
  \usepackage{lscape} 
\usepackage{tikz}
\tikzstyle{bag} = [align=center]
\usepackage{rotating}
\usepackage{amsmath}
\usepackage{amsfonts}
\usepackage{amssymb}
\usepackage{amsthm}
\usetikzlibrary{shapes}
\usepackage{graphicx}
\usepackage{stackengine}
\usepackage{xcolor}
\usepackage{multirow}

\newtheorem{theorem}{Theorem}
\newtheorem{definition}{Definition}
\newtheorem{proposition}{Proposition}

\newtheorem{lemma}{Lemma}

\newtheorem{observation}{Observation}

\title{Comparing Social Network Dynamic Operators}
\author{
Edoardo Baccini
\institute{
University of Groningen \\
}
\email{e.baccini@rug.nl}
\and
Zoé Christoff
\institute{
University of Groningen\\}
\email{z.l.christoff@rug.nl}
}

\def\titlerunning{Comparing Social Network Dynamic Operators}
\def\authorrunning{E. Baccini \& Z. Christoff  }

\begin{document}

\maketitle

\begin{abstract}

Numerous logics 
have been developed 
to reason either about 
threshold-induced opinion
diffusion in a network, or about similarity-driven network structure 
evolution, 
or about both.  
In this paper, we first introduce a logic containing
different dynamic operators to capture changes that are `asynchronous' (opinion change only, network-link change only) and changes that are `synchronous' (both at the same time). 
Second, we show that 
synchronous operators cannot, in general, be replaced by 
asynchronous operators and vice versa.  
Third, we characterise the class of models on which the synchronous operator \emph{can} be reduced to sequences of asynchronous operators. 
\end{abstract}

\section{Introduction}

There are two main types of change affecting agents connected through a social network. 
First, the features of an agent, e.g., their opinions or behavior, can be influenced by 
its neighbors in the network: for instance, if one's entire social circle has adopted an opinion in favor of (or against) vaccines, one is unlikely to disagree with this opinion. 
Under this type of social influence, or \textit{social conformity pressure},  network-neighbors 
tend to align their opinions (or any other feature that can change) and therefore become more similar.  
Second, in addition to changing their own state (opinion, or other feature), agents can also reshape their social environment by connecting with others. What generally drives the formation of new links between two agents is their \textit{similarity}. 
Both types of changes relate to how similar agents are: social influence makes network neighbors become more similar while new links make similar agents become more connected \cite[Ch. 4]{networks}.

In social network analysis, a common way of representing both types of dynamics is to assume that certain  \emph{thresholds} drive the dynamics.
On the one hand, 
a typical way of representing social influence is via \textit{threshold models} \cite{Granovetter1978, thr-limited, dodds2009threshold, networks}: agents adopt a feature when a large enough proportion of their network neighbors has already adopted it. 
On the other hand, the formation of new links has been modelled in a similar way. In probabilistic models, it is usual to assume that agents who are more similar are more likely to connect than those who are less similar \cite{talaga2020, BRAMOULLE20121754}. In deterministic models, this has been translated by a similarity threshold: two agents get connected as soon as they are similar enough \cite{howmakefriends, smetslogical2019, Smets2019ALA, smetscloseness-2020}.


Both types of changes have been addressed in logic. Indeed, a number of logical frameworks has flourished to reason about threshold-based social influence \cite{baltagdynamic2018, christoffdiffusion2019, christofflogic2015, Seligmanetal11, Seligmanetal:synthese}, and
about threshold-based link formation \cite{howmakefriends, smetslogical2019, Smets2019ALA, smetscloseness-2020}. 
Yet, to the exception of \cite{Generaldynamicdynamiclogic, pedersensimmelian, baccinichristoffverbrugge2022}, 
either the two aspects have been treated separately \cite{solaki2016logic, howmakefriends, smetslogical2019, Smets2019ALA, smetscloseness-2020, sel-gonzales}, or the two types of changes have been taken to happen one after the other \cite{Smets2019ALA}. 
To our knowledge, only 
\cite{baccinichristoffverbrugge2022} provides a logic capturing specifically  \emph{simultaneous} changes of the network structure and the state of the agents. 

In this paper, we introduce a closely related framework that, similarly to \cite{baccinichristoffverbrugge2022} combines three dynamic operators: one corresponding to the change of the network structure only, one corresponding to the change of the agents feature (opinion/behavior/state) only, and one corresponding to \emph{both changes at the same time}, but restricting ourselves to monotonic changes. We then tackle for this monotonic setting an open question by \cite{Smets2019ALA, baccinichristoffverbrugge2022} 
in the literature:
Can different sequences of dynamic operators be reduced to one another? 

We first introduce the framework in Section~\ref{sec:logic}. We then discuss the (ir)replaceability of the three dynamic operators in Section~\ref{sec:irreplaceability}. We show in particular that our `synchronous' operator cannot always be replaced by any sequences of other operators (Theorem \ref{theorem:genericreplaceability}). 
We also show that, when it can be replaced, the sequence of operators replacing it can only be of four specific types (Theorem~\ref{theorem:necessary}). Finally, we characterize the class of models on which the synchronous operator can be replaced (Theorem \ref{theorem:characterisation}).

\section{Logic of asynchronous and synchronous network changes}\label{sec:logic}

We introduce a logic to reason about asynchronous and synchronous changes in
social networks. 
We use a propositional language (where atoms are parametrized by our sets of agents and features) extended with three dynamic operators $\triangle,\square,\bigcirc$, to capture, respectively, 
diffusion update, network update, and both updates happening simultaneously.


\begin{definition}[Syntax $\mathcal{L}$]
Let $\mathcal{A}$ be a non-empty finite set of agents, $\mathcal{F}$ be a non-empty finite set of features. Let $\Phi_{at}:=\{N_{ab}: a,b\in{\mathcal{A}}\}\cup\{f_a: f\in{\mathcal{F}}, a\in{\mathcal{A}}\}$ 
be the set of atomic formulas.
The syntax $\mathcal{L}$ is the following:
\[ \varphi:= \  N_{ab} \ | \ f_a \ | \ \neg{\varphi} \ | \ \varphi\wedge{\varphi} \ | \ \triangle{\varphi} | \ \square{\varphi} \ | \ \bigcirc{\varphi} \]
where $f\in{\mathcal{F}}$ and $a,b\in \mathcal{A}$.
\label{def:syntax}
\end{definition}

The connectors
$\lor$, $\rightarrow$ and $\leftrightarrow$ are defined as usual. 
$N_{ab}$ is read as 
`agent $a$ is an influencer of agent $b$'; 
$f_a$ 
as `agent $a$ has 
feature $f$'; 
$\triangle\varphi$ as `after a diffusion update, $\varphi$ holds'; 
 $\square\varphi$ as `after a network update, $\varphi$ holds'; 
 $\bigcirc\varphi$ as `after a synchronous update, $\varphi$ holds'.

We now introduce the models representing who is influencing whom and who has which features, and our three different types of updates.



\begin{definition}[Model $M$]
Let $\mathcal{A}$ be a non-empty finite set of agents, $\mathcal{F}$ be a non-empty finite set of features. A model $M$ over $\mathcal{A}$ and $\mathcal{F}$ is a tuple $\langle{\mathcal{N},\mathcal{V},\omega,\tau}\rangle$, where:

\begin{itemize}
    \item $\mathcal{N}\subseteq{\mathcal{A}\times{\mathcal{A}}}$ is a social influence relation;
    \item $\mathcal{V}:\mathcal{A}\longrightarrow{\mathcal{P}(\mathcal{F})}$ is a valuation function, 
    assigning to each agent a set of adopted features;
    \item ${\omega,\tau\in\mathbb{Q}}$ are two rational numbers such that ${0\leq{\omega}\leq{1}}$ and  ${0<{\tau}\leq{1}}$, interpreted, respectively, as similarity threshold and influenceability threshold.
\end{itemize}
We write $C^{\omega\tau}$ for the class of all models for given values of $\omega$ and $\tau$.
\label{def:model}
\end{definition}

We turn to defining the three types of model updates corresponding to 
our three dynamic operators. 
First, after the diffusion (only) update, the set of features each agent adopt is updated. 
Agents might start adopting new features if enough of their neighbors had already adopted them before the update. Note that, while \cite{Smets2019ALA, baccinichristoffverbrugge2022} consider updates in which agents might start abandoning previously adopted features, here we restrict ourselves to the case in which agents are not allowed to start unadopting features, similarly as in \cite{baltagdynamic2018}. 

\begin{definition}[Diffusion update -  $M_{\triangle}$]
Given a model ${M= \langle{\mathcal{N},\mathcal{V},\omega,\tau}\rangle}$, the updated model \linebreak ${M_\triangle=\langle{\mathcal{N},\mathcal{V'},\omega,\tau}\rangle}$
is
such that for any $a,b\in \mathcal{A}$ and any $f\in\mathcal{F}$: 
   
    \[f\in\mathcal{V}'(a) \textrm{ iff } \left\{\begin{array}{lr}
        f\in{\mathcal{V}(a)}, & \text{ if } N(a)=\emptyset \\
        f\in{\mathcal{V}(a)}\textrm{ or }\frac{|N_f(a)|}{|N(a)|}\geq{\tau}, & \text{otherwise}
        \end{array}\right\}\] 
        
        where $N_f(a):=\{b\in{A}: (b,a)\in{\mathcal{N}} \text{ and } f\in{\mathcal{V}(b)}\}$ and $N(a):=\{b\in{A}: (b,a)\in{\mathcal{N}}\}$.
\label{def:mdiffusion} 
\end{definition}

The diffusion update does not affect the network structure. In contrast, the network update only affects the connections, not the features adopted by any of the agents.
After a network update, new links may have formed between agents that agree on sufficiently many features. 
Just as agents could not unadopt previously adopted features, agents cannot break old connections, which differs for instance from \cite{pedersensimmelian, Smets2019ALA}. In this respect, our network update is a monotonic version of that in \cite{Smets2019ALA}.

\begin{definition}[Network update - $M_{\square}$]
Given a model ${M= \langle{\mathcal{N},\mathcal{V},\omega,\tau}\rangle}$, the updated model \linebreak ${M_\square=\langle{\mathcal{N'},\mathcal{V},\omega,\tau}\rangle}$
is
such that for any $a,b\in \mathcal{A}$ and any $f\in\mathcal{F}$: 
    \[(a,b)\in\mathcal{N}' \textrm{ iff } (a,b)\in\mathcal{N}\textrm{ or }\frac{|(\mathcal{V}(a)\cap\mathcal{V}(b))\cup(\mathcal{F}\setminus({\mathcal{V}(a)}\cup{{\mathcal{V}(b)}))|}}{|\mathcal{F}|}\geq{\omega}\]
\label{def:modelupdate} 
\end{definition}




Third, the synchronous 
update affects features and connections 
at once. Adoption of new features happens under the same conditions as with the diffusion update, and new links are created under the same conditions as with the network update.

\begin{definition}[Synchronous update - $M_{\bigcirc}$]
Let ${M={\langle{\mathcal{N},\mathcal{V},\omega,\tau}\rangle}}$, the model resulting from synchronous 
update
is 
${M_{\bigcirc}:=\langle{\mathcal{N}',\mathcal{V'},\omega,\tau}\rangle}$,  where $\mathcal{N}'$ is as in Definition~\ref{def:modelupdate} and $\mathcal{V}'$
is as in Definition~\ref{def:mdiffusion}.
\label{def:diachronic}
\end{definition}

Now that we have defined the model-updates, we can introduce the semantic clauses for formulas containing  the corresponding operators.  

\begin{definition}[Satisfaction]
For any model $M=\langle{\mathcal{N},\mathcal{V},\omega,\tau}\rangle$ and any formula $\varphi\in\mathcal{L}$, the truth of $\varphi$ in $M$ is inductively defined as follows:
\begin{multicols}{2}
\begin{itemize}
    \item[] $M\models{f_a}$ if and only if $f\in{\mathcal{V}(a)}$
    \item[] $M\models{N_{ab}}$ if and only if $(a,b)\in{\mathcal{N}}$ 
    \item[] $M\models{\neg\varphi}$ if and only if $M\not\models{\varphi}$
    \item[] $M\models{\varphi\wedge\psi}$ if and only if $M\models\varphi$ and $M\models{\psi}$
    \item[] $M\models{\square\varphi}$ if and only if $M_{\square}\models{\varphi}$
    \item[] $M\models{\triangle\varphi}$ if and only if $M_{\triangle}\models{\varphi}$
    \item[] $M\models{\bigcirc\varphi}$ if and only if $M_{\bigcirc}\models{\varphi}$

\end{itemize}
\end{multicols}
where $M_{\triangle}$ is the updated model as in Definition~\ref{def:mdiffusion}, and $M_{\square}$ is the updated model 
as in  Definition~\ref{def:modelupdate}, and $M_{\bigcirc}$ is the updated model as in Definition~\ref{def:diachronic}.
\label{def:semantics}
\end{definition}

As usual, we say that a formula is valid in a class of models if it is true in all models of that class and valid (tout court) if it is valid in all models.

\begin{observation}
    Let $M_1=\langle{\mathcal{N}_1,\mathcal{V}_1,\omega,\tau}\rangle$ and $M_2=\langle{\mathcal{N}_2,\mathcal{V}_2,\omega,\tau}\rangle$ be two models. The following are equivalent:
    \begin{itemize}
    \item for all $\varphi_{at}\in\Phi_{at}$, $M_1\models\varphi_{at}$ iff $M_2\models\varphi_{at}$
    \item for all $\varphi\in\mathcal{L}$, $M_1\models\varphi$ iff $M_2\models\varphi$
    \item $M_1=M_2$  
    \end{itemize}
    \label{obs:first}
\end{observation}

We introduce the following two abbreviations 
capturing, respectively,  when an agent $a$ has sufficient pressure to adopt a feature ($f^\tau_{N(a)}$), and when two agents $a$ and $b$ have sufficient similarity to 
connect ($sim^{\omega}_{ab}$). 


\[f^{\tau}_{N(a)
}:= {\bigvee_{\{G\subseteq{N}\subseteq{A},\textrm{ } N\neq{\emptyset} \textrm{ }: \frac{|G|}{|N|}\geq{\tau}\}}(\bigwedge_{b\in{N}}N_{ba} \wedge \bigwedge_{b\not\in{N}}\neg{N_{ba}}\wedge\bigwedge_{b\in{G}}f_b)}\] 


\[{sim^{\omega}_{ab} := \bigvee_{\{E\subseteq{\mathcal{F}}: \frac{|E|}{|\mathcal{F}|}\geq{\omega}\}}\bigwedge_{f\in{E}}(f_a\leftrightarrow{f_b})}\]



    
      

These abbreviations can then be used to obtain reduction axioms for each of the dynamic modalities in $\mathcal{L}$, which are shown in Table~\ref{table:axiom1}.
The reduction axioms for the dynamic operators in $\mathcal{L}$ are very similar to those in other dynamic logics of social network change. 
Indeed, the reduction axioms for the operator $\triangle$ are the same as the those of the dynamic operator $[adopt]$ in \cite{baltagdynamic2018}, with the exception that our logic captures multiple diffusing features and thus contains reduction axioms for each spreading feature in $\mathcal{F}$. In this sense, they resemble the reduction axioms in \cite{Smets2019ALA} with the difference that in our setting features cannot be unadopted. 
Moreover, the reduction axioms for the operator $\square$ are similar to those in \cite{Smets2019ALA}, with the difference that 
our framework does not allow for link deletion.
The reduction axioms for the operator $\bigcirc$ merely reflect the fact that both features and links are affected by a synchronous update. 

 We will investigate how and when operators can replace one another in the next section. Before that, by looking at our axioms, we can immediately observe that an operator can replace another when it precedes specific formulas: 


\begin{observation}
    Let ${M={\langle{\mathcal{N},\mathcal{V},\omega,\tau}\rangle}}$ be a model. For all $a,b\in\mathcal{A}$, for all $f\in\mathcal{F}$:
    \begin{multicols}{2}
    \begin{itemize}
        \item $M\models \bigcirc f_a$ iff $M\models \triangle f_a$
        \item if $M\models \square f_a$, then $M\models \bigcirc f_a$
        \item $M\models \bigcirc N_{ab}$ iff $M\models \square N_{ab}$
        \item if $M\models \triangle N_{ab}$ then $M\models \bigcirc N_{ab}$
        
    \end{itemize}
    \end{multicols}
    \label{observation:inclusion}
\end{observation}

\begin{table}[t!]
    \centering
\begin{tabular}{ c | c | c}
$\square{N_{ab}}\leftrightarrow{N_{ab}\lor{sim^{\omega}_{ab}}}$ & $\triangle{N_{ab}}\leftrightarrow{N_{ab}}$ &  $\bigcirc{N_{ab}}\leftrightarrow{N_{ab}\lor{sim^{\omega}_{ab}}}$\\
 $\square{f_{a}}\leftrightarrow{f_a}$ & $\triangle{f_{a}}\leftrightarrow{f_a\lor{f^\tau_{N(a)}}}$ &  $\bigcirc{f_{a}}\leftrightarrow{f_a\lor{f^\tau_{N(a)}}}$\\

${\square(\varphi\wedge{\psi})}\leftrightarrow{\square{\varphi}\wedge{\square{\psi}}}$ & ${\triangle(\varphi\wedge{\psi})}\leftrightarrow{\triangle{\varphi}\wedge{\triangle{\psi}}}$ &  ${\bigcirc(\varphi\wedge{\psi})}\leftrightarrow{\bigcirc{\varphi}\wedge{\bigcirc{\psi}}}$\\
    $\square{\neg{\varphi}}\leftrightarrow{\neg{\square{\varphi}}}$ & $\triangle{\neg{\varphi}}\leftrightarrow{\neg{\triangle{\varphi}}}$ & $\bigcirc{\neg{\varphi}}\leftrightarrow{\neg{\bigcirc{\varphi}}}$\\
    \multicolumn{3}{c}{}\\
    \multicolumn{3}{c}{From $\varphi_1\leftrightarrow\varphi_2$, infer that $\varphi\leftrightarrow\varphi[\varphi_1/\varphi_2]$, where $\varphi[\varphi_1/\varphi_2]$ is a formula }\\
    \multicolumn{3}{c}{  obtained by replacing one or more occurrences of $\varphi_1$ with $\varphi_2$}

 \end{tabular}
 \caption{Reduction Axioms and derivation rule for the dynamic modalities $\square,\triangle,\bigcirc$.
 }
 \label{table:axiom1}
\end{table}

\begin{definition}[Logic $L^{\omega\tau}$ 
] Let $\omega\in{[0,1]}$ and $\tau\in(0,1]$ be two rational numbers.
The Logic $L^{\omega\tau}$ consists of some complete axiomatisation and derivation rules of propositional logic, together with the reduction axioms and the derivation rule in Table~\ref{table:axiom1}. 
    
\end{definition}

\begin{theorem}
    Let $\omega\in{[0,1]}$ and $\tau\in(0,1]$ be two rational numbers. For any $\varphi\in\mathcal L$: $\models_{\mathcal{C}^{\omega\tau}}\varphi \textrm{ iff } \vdash_{L^{\omega\tau}}\varphi\ $
\label{theorem:completeness}
\end{theorem}


    The proof uses standard techniques and is very similar to that of the related settings in \cite{baltagdynamic2018,Smets2019ALA,baccinichristoffverbrugge2022}: a sketch 
    is included in the \hyperref[appendix]{Appendix}.

\section{Irreplaceability of synchronous operators}\label{sec:irreplaceability}

Given that our dynamic formulas are reducible to the static 
fragment of our language, the question of comparing the expressivity of fragments of our language excluding one or two of the dynamic operators is uninteresting. In contrast, what is interesting, as suggested already in \cite{Smets2019ALA, baccinichristoffverbrugge2022},
is to compare whether formulas containing some (specific combinations of) dynamic operators could be translated into formulas containing  other (combinations of) dynamic operators. Another way to put it, closer to the way \cite{Smets2019ALA} first introduces the question, is to ask when different sequences of different model updates result in the same model.   



To be able to investigate the extent to which our dynamic operators are inter-translatable or not (beyond the atomic preceding cases mentioned in Observation~\ref{observation:inclusion}), we first have to introduce some notation and define the relevant type of expressivity criteria. 

\begin{definition}[Notation for sequences of operators] Let $D=\{\bigcirc,\triangle, \square\}$. For $O\subseteq{D}$, $S_O$ denotes the set of all non-empty finite sequences of operators in 
$O$. We write $d_1d_2...d_n$ for the sequence $\langle{d_1,d_2,...,d_n}\rangle$ 
and $d^n$ for the sequence consisting of $n\in \mathbb{N}$ repetitions of $d\in D$.  
We denote by $s^{j:k}$ the subsequence of $s$ starting with the $j$-th element of $s$ and ending with the $k$-th element of $s$. 
Given two sequences $s_1, s_2$ of lengths $n,m\in\mathbb{N}$, respectively, we write $s_1s_2$ for the sequence of length $n+m$ obtained by prefixing $s_1$ to $s_2$. 
\end{definition}

\begin{definition}[Equivalence of sequences]
Two sequences $s_1,s_2
\in{S_D}$ are equivalent on a model $M$ when $M_{s_1}=M_{s_2}$, or, equivalently (by Observation \ref{obs:first}), when for all $\varphi\in\mathcal{L}$, $M\models{s_1}\varphi$ 
if and only if $M\models{s_2}\varphi$. Two sequences $s_1$ and $s_2$ are equivalent over a class of models 
when they are equivalent over all models in the class. 
Two sequences are equivalent (tout court) when they are equivalent over the class of all models.  
\label{def:equivalences}
\end{definition}

We start by making some observations about sequences of $\triangle$ and $\square$ operators.

\begin{observation}
    Let a model ${M= \langle{\mathcal{N},\mathcal{V},\omega,\tau}\rangle}$ be given.
    \begin{itemize}
        \item Any sequence $s\in{S_{\{\square\}}}$ is equivalent to the sequence $\square$ on $M$.
        \item There exists an $n<{|\mathcal{A}|}$, such that, for any $m>n$, $\triangle^m$ is equivalent to $\triangle^n$ on $M$.
    \end{itemize}
 \label{lemma:bound}
\end{observation}

    The first point follows from the fact that the model update in Definition~\ref{def:modelupdate} is idempotent, and therefore $M\models\square^n \varphi$ if and only if $M\models\square \varphi$.
    A proof of the second point can be found in \cite{baltagdynamic2018}.

We then lift this notion of equivalence between sequences to an existential notion between sets of sequences, so that we can compare the different dynamic fragments of our language.

\begin{definition}[Replaceability of sets]
Let $S_1,S_2
\subseteq{S_D}$
be two sets of sequences. 
The set $S_1$ is replaceable with the set $S_2$ in a model $M$, when, for all sequences $s_1\in{S_1}$, there exists 
a sequence $s_2\in{S_2}$ that is equivalent to $s_1$ in $M$.
$S_1$ is replaceable with $S_2$ over a class of model 
when it is repleaceable with $S_2$ in all models of the class.
$S_1$ is replaceable with $S_2$ (tout court) when it is repleaceble with $S_2$ over the class of all models.



\label{replaceability}
\end{definition}

When comparing our dynamic operators, it is easy to see that $S_{\{\square\}}$ (and therefore any superset of it) is not replaceable by $S_{\{\triangle\}}$ and, vice versa, that $S_{\{\triangle\}}$ (and therefore any superset of it) is not replaceable by $S_{\{\square\}}$ and similarly for $S_{\{\square\}}$ and $S_{\{\bigcirc\}}$, and $S_{\{\triangle\}}$ and $S_{\{\bigcirc\}}$, which  implies that $S_{\{\square,\triangle\}}$ is not replaceable with $S_{\{\bigcirc\}}$. The only interesting question is: can  we replace our synchronous operator?
\begin{theorem}
    $S_{\{\bigcirc\}}$ is not replaceable with $S_{\{\square,\triangle\}}$.
    \label{theorem:genericreplaceability}
\end{theorem}

\begin{figure}
    \centering
    \includegraphics[trim={2 1 1 1}, clip, width=0.85\linewidth]{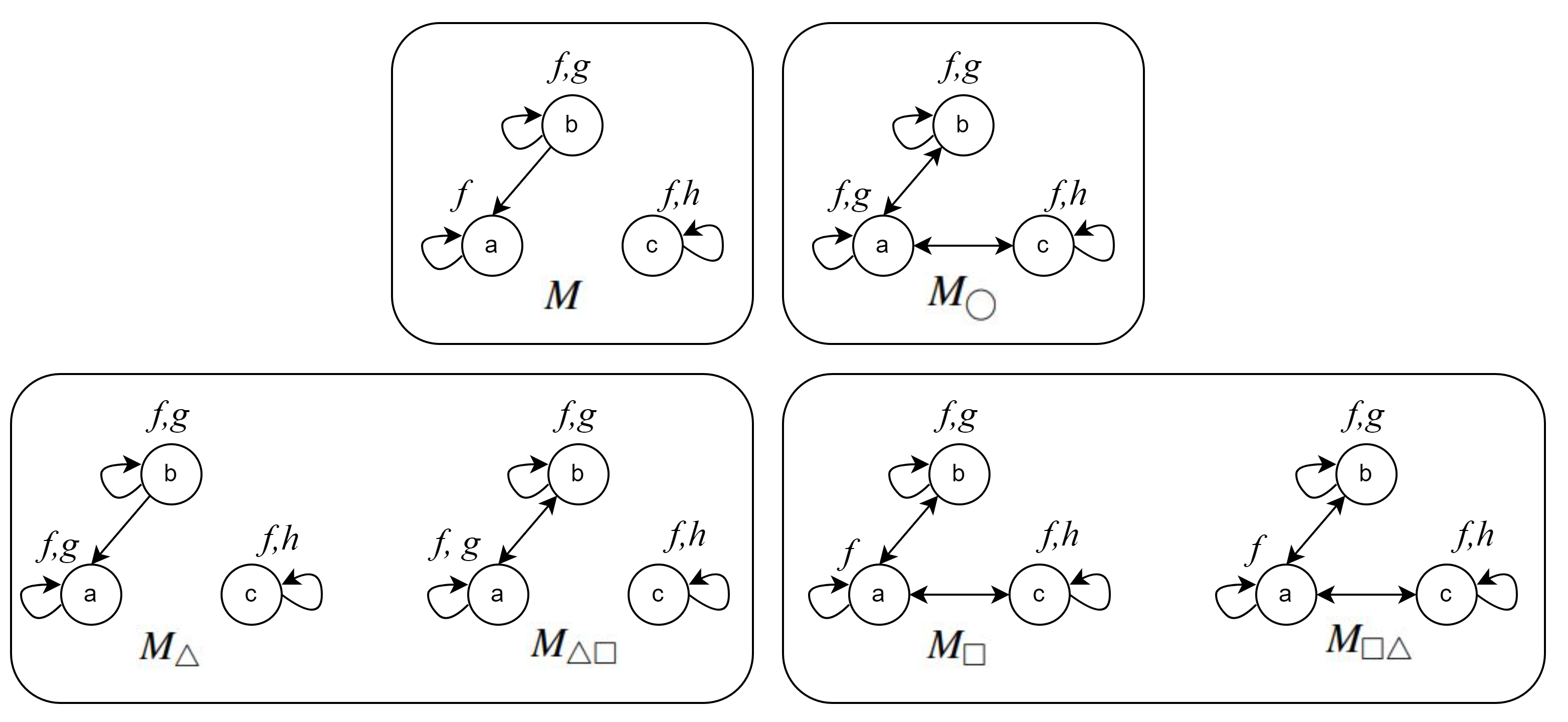}
    \caption{The figure represents the model $M$, and its model updates $M_\bigcirc$, $M_\triangle$, $M_\square$, $M_{\triangle\square}$, $M_{\square\triangle}$ considered in Theorem~\ref{theorem:genericreplaceability}. The model $M$ is as follows: $\mathcal{A}=\{a,b,c\}$, $\mathcal{F}=\{f,g,h\}$, $\omega=\tau=\frac{1}{2}$. For each model, each agent is represented as a node, and the influence of one agent on another agent is represented as a directed arrow from the influencing node to the influenced node. Next to each node, all and only the features in $\mathcal{F}$ that an agent possess are reported.}
    \label{fig:fig1}
\end{figure}

\begin{proof}

     \noindent We show that there is no sequence in $S_{\{\triangle,\square\}}$ that is equivalent to the sequence $\bigcirc$ on all models. 
Assume, towards a contradiction that there exists a sequence $s\in{S_{\{\triangle,\square\}}}$ equivalent to $\bigcirc$ in the model $M$ given in Fig.~\ref{fig:fig1}. 
Let $n\in{\mathbb{N}}$ be the length of $s$. One of two cases must hold:
\begin{itemize}
    \item[] [Case 1: $s$ starts with $\triangle$.]
We 
can rewrite $s$ as $\triangle{s^{2:n}}$. 
From Fig.\ref{fig:fig1}, we know that $M\not\models\triangle{{{N_{ac}}}}$, whereas $M\models\bigcirc{{{N_{ac}}}}$ and therefore $s\neq \triangle$. 
Note that $M$ is such that any number of successive triangles reduce to one: for any $\varphi$, $M\models\triangle\varphi$ iff $M\models\triangle{s'}\varphi$ for every ${s'}\in{S_{\{\triangle\}}}$. Since $M\not\models\triangle{{{N_{ac}}}}$, then $M\not\models\triangle{s'}{N_{ac}}$, for every ${s'}\in{S_{\{\triangle\}}}$. Thus,  $s\not\in{S_{\{\triangle\}}}$.
The sequence $s$ must therefore contain at least one $\square$, and can be rewritten as ${s^{1:m}\square{s^{(m+2):n}}}$ where $s^{1:m}\in{S_{\{\triangle\}}}$, with $1\leq{m}<n$. 
Given that, as mentioned above, all triangles can be reduced to one, for all $\varphi\in{\mathcal{L}}$, $M\models{{s^{1:m}\square{s^{(m+2):n}}}\varphi}$ iff $M\models{{\triangle\square{s^{(m+2):n}}}\varphi}$. 
As illustrated in Figure \ref{fig:fig1}, $M\not\models{{\triangle\square}N_{ac}}$. Furthermore, note that $M\models{{\triangle\square}\varphi}$ iff  $M\models{{\triangle\square}s'\varphi}$ for all $s'\in S_{\{\triangle, \square\}}$, i.e., $M_{\triangle\square}$ is stable. 
 Hence, $M\not\models{{\triangle\square}s'{N_{ac}}}$ for all $s'\in S_{\{\triangle, \square\}}$. Thus, in particular: $M\not\models{{\triangle\square}s^{(m+2):n}{N_{ac}}}$. Hence, using again the fact that the number of initial triangles is irrelevant, $M\not\models{s^{1:m}\square{s^{(m+2):n}}}{N_{ac}}$ which is $M\not\models{s}{N_{ac}}$. But $M\models{\bigcirc}{N_{ac}}$, and hence $s$ is not equivalent to $\bigcirc$ in $M$.  
\item[] [Case 2: $s$ starts with $\square$.]
$s$ is of the kind ${\square{s^{2:n}}}$. As illustrated in Fig.\ref{fig:fig1}, $M\not\models\square{{{g_a}}}$, whereas $M\models\bigcirc{{{g_{a}}}}$.
Note that $M$ is such that $M\models{{\square}\varphi}$ iff  $M\models{{\square}s'\varphi}$ for all $s'\in S_{\{\square,\triangle\}}$. Hence, $M\not\models{{\square}s'{g_a}}$ for all $s'\in S_{\{\square,\triangle\}}$. Thus, $M\not\models{{\square}s^{2:n}{g_a}}$ which is the same as $M\not\models{s{g_a}}$. Hence, $s$ is not equivalent to $\bigcirc$ in $M$.  

\end{itemize}

Hence, 
in both cases, $s$ is not equivalent to $\bigcirc$ in $M$. Contradiction. Thus, there is no sequence in $S_{\{\triangle, \square\}}$ equivalent to $\bigcirc$, which implies that $S_{\{\bigcirc\}}$ is not replaceable by $S_{\{\triangle, \square\}}$.



\end{proof}

From the proof of 
Theorem~\ref{theorem:genericreplaceability}, we know that
there is no sequence in 
 $S_{\{\square,\triangle\}}$ that is equivalent to   
 $\bigcirc$ 
 in the class of \emph{all} models. 
However, we will show in Proposition~ \ref{prop:equivalencemonotonic} that there are classes of models on which 
$\bigcirc$ 
does have 
equivalent sequences in $S_{\{\triangle, \square \}}$.
We first need to introduce the following additional abbreviations, where we already name $\psi_s$ the formula that captures the conditions under which $\bigcirc$ is equivalent to $s$.

\begin{definition}[Abbreviations $\psi_\triangle$, $\psi_\square$, $\psi_{\triangle\square}$ and $\psi_{\square\triangle^n}$ ]
\
\begin{itemize}
    \item $\psi_\triangle:=\bigwedge_{a,b\in \mathcal{A}}(N_{ab}\lor\neg sim^\omega_{ab})$
    \item $\psi_\square:=\bigwedge_{a\in \mathcal{A}}\bigwedge_{f\in F}(f_a\lor\neg f^\tau_{N(a)})$
    \item $\psi_{\triangle\square}:=\bigwedge_{{a,b}\in \mathcal{A}}(\neg{N_{ab}}\rightarrow({sim^{\omega}_{ab}}\leftrightarrow\triangle{{sim^\omega_{ab}})}) $
    \item For all $n>0$:  $\psi_{\square\triangle^n}:=\bigwedge_{a\in \mathcal{A}}\bigwedge_{f\in{F}}(\neg{f_a}\rightarrow({f^\tau_{N(a)}}\leftrightarrow\bigvee_{0\leq{i}\leq{n-1}}{\square\triangle^{i}{f^\tau_{N(a)}})})$
\end{itemize}
\label{def:psiabbreviations}
\end{definition}

We can now show that these formulas indeed define four classes of models in which $\bigcirc$ has an equivalent sequence in $S_{\{\triangle\square\}}$:
\begin{proposition}
Let $M$ 
be a model. $\bigcirc$ is equivalent on $M$ to: 
\begin{multicols}{2}
    
\begin{itemize}
    \item $\triangle$ iff $M\models\psi_\triangle$
    \item $\square$ iff $M\models\psi_\square$
    \item $\triangle\square $ iff $M\models\psi_{\triangle\square}$
\item 
$\square\triangle^{n}$ iff $M\models\psi_{\square\triangle^n}$, for $n>0$
\end{itemize}
\end{multicols}
\label{prop:equivalencemonotonic}
\end{proposition}






For space reasons,  
the proof of Proposition~\ref{prop:equivalencemonotonic} is provided in the \hyperref[appendix]{Appendix}. 

We can now show that if $\bigcirc$ has an equivalent sequence in $S_{\{\triangle\square\}}$ on some model, then that sequence has to be equivalent to one of those in Proposition~\ref{prop:equivalencemonotonic} on that model.
To prove this,  
we need the following lemmas.

\begin{lemma}
     Let 
     $M$
     be a model and $s\in{S_D}$. 
     If $s$ starts with a subsequence of the 
     form $\triangle^n$ for some $n>0$
     and $s$ is equivalent to $\bigcirc$ on $M$, then 
 $\triangle^n$ 
  is equivalent to $\triangle$ on $M$.
    \label{lemma:sublemma}
\end{lemma}

\begin{proof}
     Consider 
     a sequence $s\in{S_D}$ such that $s$ 
     starts with a subsequence of the kind $\triangle^n$, for some $n>0$, and $s$ is equivalent to $\bigcirc$ on model ${M= \langle{\mathcal{N},\mathcal{V},\omega,\tau}\rangle}$. 
  Two cases: either 
  $n=1$, or 
 $n>1$. 
  If $n=1$, 
  the claim is trivially true since $\triangle$ is equivalent to itself.
     Assume now that 
     $n>1$. Assume also, towards
      a contradiction, that 
 $\triangle^n$ 
 is not equivalent to $\triangle$ on $M$. By Def.~\ref{def:equivalences}, it follows that $M_{\triangle}\neq M_{\triangle^n}$. By Obs.~\ref{obs:first}, it follows that $M_{\triangle}$ and $ M_{\triangle^n}$ must differ on whether they satisfy some atomic proposition.  
By Def.~\ref{def:mdiffusion} and Def.~\ref{def:modelupdate}, since features can never be abandoned, we know that, for all $a\in \mathcal{A}$ for all $f\in\mathcal{F}$, if $M\models\triangle f_a$, then $M\models\triangle s' f_a$ for all $s'\in{S_D}$; with a similar reasoning, from Def.~\ref{def:mdiffusion} and Def.~\ref{def:modelupdate}, since diffusion updates do not alter the network structure, we also know that for all $a,b\in\mathcal{A}$, $M\models\triangle N_{ab}$ iff $M\models\triangle^n N_{ab}$.  
These observations, together with the fact that $M_{\triangle^n}$ and the model $M_\triangle$ must differ on the satisfaction of some atomic proposition,
imply that there are $a\in\mathcal{A}$ and $f\in\mathcal{F}$ such that $M\models \triangle^n f_a$ and $M\not\models \triangle f_a$. 
 From Obs.~\ref{observation:inclusion}, we know that, for all $a\in \mathcal{A}$ for all $f\in\mathcal{F}$, $M\models\triangle f_a$ iff $M\models\bigcirc{f_a}$. Since $M\not\models\triangle f_a$, we can infer that $M\not\models \bigcirc f_a$. 
At the same time, from the fact that features can never be abandoned, and the fact that $M\models\triangle^n f_a$, it follows that $M\models s f_a$, since $\triangle^n$ 
is the initial subsequence of $s$. 
Therefore, 
it must be the case that 
both $M\not\models\bigcirc f_a$ and that $M\models s f_a$. This contradicts the initial assumption that $s$ and $\bigcirc$ are equivalent on $M$. 
Therefore, for all $n\geq 1$, 
 $\triangle^n$ 
 must be equivalent to $\triangle$ on $M$.
\end{proof}

\begin{lemma}
    Let 
    $M$ 
    be a a model 
and $s\in{S_{\{\triangle,\square\}} \setminus(S_{\{\triangle\}}\cup{S_{\{\square\}}})}$. 
If $s$ starts with $\triangle$ and is equivalent to $\bigcirc$ on $M$, then $s$ is equivalent to $\triangle\square$ on $M$.
    \label{lemma:starttriangle}
\end{lemma}

\begin{proof}
     Consider any  $s\in{S_{\{\triangle,\square\}}\setminus(S_{\{\triangle\}}}\cup{S_{\{\square\}}})$, such that $s$ starts with $\triangle$ and is equivalent to $\bigcirc$ on some model 
     ${M= \langle{\mathcal{N},\mathcal{V},\omega,\tau}\rangle}$.
     Assume, towards contradiction, that $s $  is not equivalent to $\triangle\square$ on $M$. 
    By Lemma \ref{lemma:sublemma}, we know that if $s$ starts with a sequence of the kind $\triangle^n$, then 
 $\triangle^n$ must be equivalent to $\triangle$ on $M$. For this reason, we can restrict ourselves to consider the case in which $s$ starts with the subsequence $\triangle\square$.  
Since $s$ is not equivalent to $\triangle\square$ on $M$, by Obs.\ref{obs:first}, it follows that 
 $M_{s}$ and $M_{\triangle\square}$ must differ on whether they satisfy some atomic proposition.
From Obs.~\ref{observation:inclusion}, we know that, for all $a\in \mathcal{A}$ for all $f\in\mathcal{F}$, $M\models\triangle f_a$ iff $M\models\bigcirc{f_a}$.
From this and the fact that the network update does not affect the features of the agent, we know that $M\models\triangle\square f_a$ iff $M\models\bigcirc{f_a}$, for all $a\in \mathcal{A}$ for all $f\in\mathcal{F}$. 
This, combined with the fact that $s$ is equivalent to $\bigcirc$, implies that $M\models\bigcirc{f_a}$ iff $M\models sf_a$, and hence, that $M\models\triangle\square f_a$ iff $M\models sf_a$, for all $a\in \mathcal{A}$ and $f\in\mathcal{F}$.
Since by Obs.\ref{obs:first}, we know that $M_{s}$ and $M_{\triangle\square}$ must differ on whether they satisfy some atomic proposition,
it follows that there are $a,b$ such 
that either: (i) $M\not\models s{N_{ab}}$ and $M\models \triangle\square{N_{ab}}$ or (ii) $M\models s{N_{ab}}$ and $M\not\models \triangle\square{N_{ab}}$. 
Assume that (i) is the case. If $M\models \triangle\square{N_{ab}}$, then from the Def.~\ref{def:mdiffusion} and Def.~\ref{def:modelupdate}, it follows that $M\models\triangle\square s'{N_{ab}}$ for all $s'\in S_{D}$, since links cannot be deleted. This fact, together with the fact that $\square\triangle$ is a subsequence of $s$, implies in particular that $M\models sN_{ab}$. This contradicts the fact that $M\not\models s{N_{ab}}$. 
Therefore (ii) 
must be the case, 
 i.e. $M\models s{N_{ab}}$ and $M\not\models \triangle\square{N_{ab}}$.
By the assumption that $s$ is equivalent to $\bigcirc$ it follows that $M\models{\bigcirc N_{ab}}$. 
Now, by the reduction axioms and the fact that $M\not\models \triangle\square{N_{ab}}$, we know that $M\not\models N_{ab}$ and $M\not\models\triangle{sim^\omega_{ab}}$. 
Since $M\not\models\triangle{sim^\omega_{ab}}$ but $M\models{s N_{ab}}$, there must $m\in\mathbb{N}$ smaller than the length $n$ of the sequence $s$, such that $M\models s^{1:m} sim^\omega_{ab}$: in other words, it must be the case that at some point of the sequence $s$, the agents $a,b$ have become similar. The fact that $M\models s^{1:m} sim^\omega_{ab}$ holds 
implies that there exist at least one $f\in{\mathcal{F}}$ and a $1<j\leq{m}$ such that either $M\models{s^{1:j}f_a}$ and $M\not\models{s^{1:i}f_a}$ for all $i<{j}$, or $M\models{s^{1:j}f_b}$ and $M\not\models{s^{1:i}f_b}$ for all $i<{j}$: this simply means that in order to become similar, at least one among $a$ or $b$ must have acquired at least one new feature that makes them similar at some point in the update sequence expressed by $s$. 
W.l.o.g. consider the case in which it is $a$ that has acquired a new feature, i.e. that there exist $f\in{\mathcal{F}}$ and a $1<j\leq{m}$ such that  $M\models{s^{1:j}f_a}$ and $M\not\models{s^{1:i}f_a}$ for all $i<{j}$. From the fact that features are never abandoned, $M\models sf_a$. 
Since $M\models sf_a$, and, by assumption $s$ is equivalent to $\bigcirc$ on $M$, it follows that $M\models\bigcirc{f_a}$. 
Since, by assumption $s$ starts with $\triangle$, and it is the case that $M\not\models{s^{1:i}f_a}$ for all $i<{j}$, we know that $M\not\models\triangle{f_a}$ (triangle is the first operator in $s$). 
It follows that both $M\models\bigcirc{f_a}$ and $M\not\models\triangle{f_a}$ are true. By Obs.~\ref{observation:inclusion},
we know that for all $a\in \mathcal{A}$ for all $f\in\mathcal{F}$, $M\models\bigcirc{f_a}$ iff $M\models\triangle{f_a}$. Contradiction.
Therefore, there is no sequence $s\in{S_{\{\triangle,\square\}}\setminus(S_{\{\triangle\}}}\cup{S_{\{\square\}}})$, such that $s$ starts with $\triangle$, is equivalent to $\bigcirc$ on $M$, and is not equivalent to $\triangle\square$ on $M$.

\end{proof}


\begin{lemma}
   Let $M$ be a model 
   and $s\in{S_{\{\triangle,\square\}}\setminus(S_{\{\triangle\}}}\cup S_{\{\square\}})$. If $s$ starts with $\square$ and is equivalent to $\bigcirc$ on $M$, then $s$ is equivalent to a sequence in the set $\{\square\triangle^n : n>0\}$.
    \label{lemma:startsquare}
\end{lemma}

\begin{proof}
    Let a model ${M= \langle{\mathcal{N},\mathcal{V},\omega,\tau}\rangle}$ be given. Assume that $s\in{S_{\{\triangle,\square\}}\setminus(S_{\{\triangle\}}}\cup{S_{\{\square\}}})$ starts with $\square$ and is equivalent to $\bigcirc$ on $M$, and that $s$ is not equivalent to any sequence in the set $\{\square\triangle^n : n>0\}$ on $M$.
    From the fact that any sequence $s'\in{S_{\{\square\}}}$ is equivalent to the sequence $\square$, it follows that, if $s$ starts with a subsequence of the kind $\square^n$ before the first occurrence of a $\triangle$, then $s$ is equivalent on $M$ to a sequence that starts with a single $\square$ followed by the subsequence of $s$ starting at the first occurrence of a $\triangle$ and ending with the last operator of $s$. In other words, it is sufficient to consider the case in which $s$ starts with a subsequence of the kind $\square\triangle$.
From this, and the fact that $s\in{S_{\{\triangle,\square\}}\setminus(S_{\{\triangle\}}}\cup{S_{\{\square\}}})$ and $s\not\in\{\square\triangle^n : n>0\}$, $s$ must be such that at some point of the subsequence of $s$ starting with the third operator of $s$, at least another $\square$ occurs in it. 
Furthermore, at least one such $\square$, must be such that $s$ is not equivalent on $M$ to the sequence $s$ without that $\square$. Otherwise, $s$ would be equivalent to a sequence with no further elements of the kind $\square$ after the initial subsequence $\square\triangle$, and hence would be a sequence in the set $\{\square\triangle^n: n>0\}$.
From this, it follows that there must be a subsequence $s^{1:m}$ of $s$, with $m\leq{n}$, with $n$ the length of the sequence $s$, where the $m$-th element is a $\square$, such that for some $a,b\in\mathcal{A}$, $M\models{s^{1:m}N_{ab}}$, and for all $j<{m}$, $M\not\models{s^{1:j}N_{ab}}$.
 Observe that $\square$ is a subsequence of $s^{1:m}$. Therefore, $M\not\models\square{N_{ab}}$. 
 By the fact that $M\models{s^{1:m}N_{ab}}$, and since ${s^{1:m}}$ is a subsequence of $s$, and connections between agents cannot be abandoned by Def.~\ref{def:modelupdate}, it follows that $M\models{s N_{ab}}$.
 By the assumption that $s$ and $\bigcirc$ are equivalent, it follows that $M\models{s N_{ab}}$.
By Obs.~\ref{observation:inclusion}, 
we know that for all $a,b\in \mathcal{A}$, $M\models\bigcirc N_{ab}$ iff $M\models\square N_{ab}$. This contradicts the previous claim that $M\not\models\square{N_{ab}}$.
Therefore there is no sequence $s\in{S_{\{\triangle,\square\}}\setminus(S_{\{\triangle\}}}\cup{S_{\{\square\}}})$ that starts with $\square$, is equivalent to $\bigcirc$ on $M$, and is not equivalent to any sequence in the set $\{\square\triangle^n : n>0\}$ on $M$.
\end{proof}

We can now combine the above lemmas to prove the following theorem.

\begin{theorem}
Let $M$ be a model 
and $s \in 
 S_{\{\square,\triangle\}}$.
If $s$ is equivalent to $\bigcirc$ on $M$, then $s$ is equivalent to a sequence in the set $\{\square, \triangle, \triangle\square\}\cup\{\square\triangle^n : n>0\}$ on $M$.
\label{theorem:necessary}
\end{theorem}

\begin{proof}
Consider an arbitrary model ${M= \langle{\mathcal{N},\mathcal{V},\omega,\tau}\rangle}$, and an arbitrary sequence in $s\in{S_{\{\square,\triangle\}}}$ equivalent to $\bigcirc$ on $M$.
One of three cases must hold: $s\in{S_{\{\square\}}}$, $s\in{S_{\{\triangle\}}}$, $s\in{S_{\{\triangle, \square\}}}\setminus ({S_{\{\square\}}}\cup {S_{\{\triangle\}}}) $.
[Case 1:  $s\in{S_{\{\square\}}}$.]
Since the model update in Def.~\ref{def:modelupdate} is idempotent, any sequence $s'\in{S_{\{\square\}}}$ is equivalent to the sequence $\square$. 
This implies that $s$ is equivalent to $\square$ on $M$.  
[Case 2: 
$s\in{S_{\{\triangle\}}}$.] 
Then, trivially, $s$ starts with a subsequence of the 
     form $\triangle^n$ for some $n>0$. 
From Lemma~\ref{lemma:sublemma}, we know that $s$ is equivalent to $\triangle$ on $M$.
[Case 3: 
$s\in{S_{\{\triangle, \square\}}}\setminus ({S_{\{\square\}}}\cup {S_{\{\triangle\}}}) $.] If $s$ starts with a $\triangle$, we know by Lemma~\ref{lemma:starttriangle} that $s$ is equivalent to $\triangle\square$ on $M$. If $s$ starts with a $\square$, by Lemma~\ref{lemma:startsquare}, we know that $s$ is equivalent on $M$ to a sequence in the set $\{\square, \triangle, \triangle\square\}\cup\{\square\triangle^n : n>0\}$.
From this, it follows that if $s$ is equivalent to $\bigcirc$ on some model $M$, then $s$ is equivalent on $M$ to a sequence in the set $\{\square, \triangle, \triangle\square\}\cup\{\square\triangle^n : n>0\}$.

\end{proof}

Using Observation~\ref{lemma:bound}, Proposition~\ref{prop:equivalencemonotonic} and Theorem~\ref{theorem:necessary}, we can now characterise the class of models on which $\bigcirc$ can be replaced by 
$S_{\{\triangle,\square\}}$.





\begin{theorem} 
$\bigcirc$ is replaceable by $S_{\{\triangle,\square\}}$ on a model $M$ iff $M\models\psi_\triangle\lor\psi_\square\lor\psi_{\triangle\square}\lor\bigvee_{0\leq{n}<{|\mathcal{A}}|}\psi_{\square\triangle^n}$.
\label{theorem:characterisation}
\end{theorem}

\begin{proof}
   Consider an arbitrary model $M$. 
    
    [$\Rightarrow$] Assume that $\bigcirc$ is replaceable with $S_{\{\triangle,\square\}}$ on $M$:
 there exists a sequence $s\in S_{\{\triangle,\square\}}$ equivalent to $\bigcirc$ on $M$. By Theorem~\ref{theorem:necessary}, we know that $s$ is equivalent to a sequence $s'\in\{\square, \triangle, \triangle\square\}\cup\{\square\triangle^n : n>0\}$. We distinguish four cases: 
 (i) $s$ is equivalent to $\triangle$ on $M$, (ii) $s$ is equivalent to $\square$ on $M$; (iii) $s$ is equivalent to $\triangle\square$ on $M$; 
 (iv) $s$ is equivalent to a sequence in $\{\square\triangle^n : n>0\}$ on $M$.
    Assume that (i). 
    By the first point in Prop.~\ref{prop:equivalencemonotonic}, 
    we know that 
$M\models\psi_\triangle$. From this, it follows that $M\models{\psi_\triangle\lor\psi_\square\lor\psi_{\triangle\square}\lor\bigvee_{0\leq{n}<{|\mathcal{A}}|}\psi_{\square\triangle^n}}$ for any $n$.
    Assume that (ii) is the case; by the second point in Prop.~\ref{prop:equivalencemonotonic}, and the fact that $s$ is equivalent to $\square$ on $M$, we know that $s$ is equivalent to $\bigcirc$ on $M$ iff $M\models\psi_\square$. From this, it follows that $M\models{\psi_\triangle\lor\psi_\square\lor\psi_{\triangle\square}\lor\bigvee_{0\leq{n}<{|\mathcal{A}}|}\psi_{\square\triangle^n}}$ for any $n$.
    Assume that (iii) is the case; by the third point in Prop.~\ref{prop:equivalencemonotonic}, and the fact that $s$ is equivalent to $\triangle\square$ on $M$, we know that $s$ is equivalent to $\bigcirc$  on $M$ iff $M\models\psi_{\triangle\square}$. From this, it follows that $M\models{\psi_\triangle\lor\psi_\square\lor\psi_{\triangle\square}\lor\bigvee_{0\leq{n}<{|\mathcal{A}}|}\psi_{\square\triangle^n}}$ for any $n$.
    Assume that (iv) is the case, and assume that $s$ is equivalent to a sequence of the kind $\square\triangle^n$ on $M$, for arbitrary $n>0$. By Obs.~\ref{lemma:bound}, we know that for all sequences $\square\triangle^n$ with $n\geq|\mathcal{A}|$, there exist an equivalent sequence $\square\triangle^m$ on $M$ such that $m<|\mathcal{A}|$. We therefore consider the case in which $n<|\mathcal{A}|$: in this case, by the fourth point in Prop.~\ref{prop:equivalencemonotonic}, it then follows that $s$ is equivalent to $\bigcirc$ on $M$ iff $M\models\psi_{\square\triangle^n}$. From this, it follows that $M\models{\psi_\triangle\lor\psi_\square\lor\psi_{\triangle\square}\lor\bigvee_{0\leq{n}<{|\mathcal{A}}|}\psi_{\square\triangle^n}}$. Since $n>0$ was arbitrary, this holds for all $n>0$ in $\mathbb{N}$.

    [$\Leftarrow$] Assume that $M\models{\psi_\triangle\lor\psi_\square\lor\psi_{\triangle\square}\lor\bigvee_{0\leq{n}<{|\mathcal{A}}|}\psi_{\square\triangle^n}}$. Therefore, one of the following four cases must hold: (i) $M\models\psi_\triangle$; (ii) $M\models\psi_\square$; (iii) $M\models\psi_{\triangle\square}$; (iv) $M\models{\bigvee_{0\leq{n}<{|\mathcal{A}}|}\psi_{\square\triangle^n}}$.
    (i) Assume that $M\models\psi_\triangle$. By Prop.~\ref{prop:equivalencemonotonic}, we know that this holds iff $\triangle$ is equivalent to $\bigcirc$ on $M$. In this case, we take $s$ to be $\triangle$.
     (ii) Assume that $M\models\psi_\square$. By Prop.~\ref{prop:equivalencemonotonic}, we know that this holds iff $\square$ is equivalent to $\bigcirc$ on $M$. In this case, we take $s$ to be $\square$.
      (iii) Assume that $M\models\psi_{\triangle\square}$. By Prop.~\ref{prop:equivalencemonotonic}, we know that this holds iff $\triangle\square$ is equivalent to $\bigcirc$ on $M$. In this case, we take $s$ to be $\triangle\square$.
      (iv) Assume that $M\models{\bigvee_{0\leq{n}<{|\mathcal{A}}|}\psi_{\square\triangle^n}}$. Therefore there is an $n<|\mathcal{A}|$, such that $M\models{\psi_{\square\triangle^n}}$.By Prop.~\ref{prop:equivalencemonotonic}, we know that this holds iff $\square\triangle^n$ is equivalent to $\bigcirc$ on $M$. In this case, we take $s$ to be $\square\triangle^n$.

      Since in all cases (i)-(iv), we can find a sequence $s\in S_{\{\triangle,\square\}}$ equivalent to $\bigcirc$ on $M$, we have proven that, if $M\models{\psi_\triangle\lor\psi_\square\lor\psi_{\triangle\square}\lor\bigvee_{0\leq{n}<{|\mathcal{A}}|}\psi_{\square\triangle^n}}$, $\bigcirc$ is replaceable with $S_{\{\triangle,\square\}}$ on $M$.
    
\end{proof}

Informally, Theorem~\ref{theorem:characterisation} tells us that a sequence $s$ without synchronous operators can
replace a synchronous 
operator only under one of the following circumstances: no agent has social pressure to adopt new features ($s$ is equivalent to $\square$); no agent is similar to any disconnected agent ($s$ is equivalent to $\triangle$); conforming to social pressure preserves similarity with disconnected agents ($s$ is equivalent to $\triangle\square$); creating new connections with similar agents does not forbid conforming to old social pressures ($s$ is equivalent to a sequence in $\{\square\triangle^n : n>0\}$). 

\begin{proposition}\label{prop:scalingup}
Let $M$ be a model. If $M\models\bigwedge_{0\leq{i}\leq (m-1)}\bigcirc^i({\psi_\triangle\lor\psi_\square\lor\psi_{\triangle\square}\lor\bigvee_{0\leq{n}<{|\mathcal{A}}|}\psi_{\square\triangle^n}})$, then $\bigcirc^m$ is replaceable by $S_{\{\triangle,\square\}}$ on $M$.    
\end{proposition}
\begin{proof}
Assume $M\models\bigwedge_{0\leq{i}\leq (m-1)}\bigcirc^i({\psi_\triangle\lor\psi_\square\lor\psi_{\triangle\square}\lor\bigvee_{0\leq{n}<{|\mathcal{A}}|}\psi_{\square\triangle^n}})$. 
Then, 
for all $i\geq 0\leq (m-1)$: $M_{\bigcirc^i}\models \psi_\triangle\lor\psi_\square\lor\psi_{\triangle\square}\lor\bigvee_{0\leq{n}<{|\mathcal{A}}|}\psi_{\square\triangle^n}$, and
therefore, by Theorem~\ref{theorem:characterisation},  
$\bigcirc$ is replaceable by $S_{\{\triangle,\square\}}$ on $M_{\bigcirc^i}$. Let $s_i$ be a sequence that replaces  $\bigcirc$ on $M_{\bigcirc^i}$.
Then the sequence $s_0...s_i...s_{m-1} $ replaces $\bigcirc^m$ on $M$.




\end{proof}



\section{Conclusion}

We have introduced a logical framework containing dynamic operators to reason about asynchronous as well as synchronous 
threshold-induced monotonic changes in social networks. We 
 showed that, in general, our synchronous operator 
 cannot be replaced 
(Theorem~\ref{theorem:genericreplaceability}), and that, on the models on which it can be replaced, only sequences of four specific types can replace it
(Theorem~\ref{theorem:necessary}).
Finally, we characterised the class of models on which the synchronous operator can be replaced (Theorem~\ref{theorem:characterisation}).  

The two most natural continuations of this work would be, first,  
to characterise 
the models on which 
sequences of (more than one) synchronous operators can be replaced,  
and, second, 
to study the replaceability of synchronous operators in the non-monotonic frameworks from  \cite{Smets2019ALA, baccinichristoffverbrugge2022}. 

Furthermore, it would be interesting to study the replaceability of richer 
operators studied in 
 epistemic/doxastic  settings such as \cite{baltagdynamic2018, howmakefriends, Seligmanetal11, Seligmanetal:synthese, sel-gonzales}, for instance 
the network announcements in \cite{Seligmanetal11, Seligmanetal:synthese, sel-gonzales} or the message passing updates in \cite{sel-gonzales}.
In this direction, 
 we could compare 
which models different updates can reach, as done in \cite{BaltagSmetsMerge,baltagdynamic2018}. 
In particular, 
it would be interesting to investigate 
what types of group knowledge 
are 
reachable
by different social network dynamic updates. 


\section*{Acknowledgments}
Zoé Christoff acknowledges support from the project Social Networks and Democracy (VENI project number Vl.Veni.201F.032) financed by the Netherlands
Organisation for Scientific Research (NWO). 

\bibliographystyle{eptcs}
\bibliography{name}

\section*{Appendix}\label{appendix}

\subsection*{Proof sketch of Theorem~\ref{theorem:completeness}}

\begin{proof}

The proof uses standard methods.

[\textit{Soundness}] The soundness of the reduction axioms for the dynamic operators $\triangle,\square,\bigcirc$ follows from the fact that they spell out the model updates in Def.~\ref{def:mdiffusion}, Def.~\ref{def:modelupdate} and Def.~\ref{def:diachronic} respectively. The soundness of the axioms $\triangle N_{ab}\leftrightarrow N_{ab}$ and $\square f_a \leftrightarrow f_a$ follows from the fact that the model update $M_\triangle$ does not alter the connections between agents (Def.~\ref{def:mdiffusion}), and the fact that, respectively, the model update $M_\square$ does not alter the features of the agents (Def.~\ref{def:modelupdate}). The soundness of the axioms $\triangle f_a \leftrightarrow f_a \lor f^\tau_{N(a)}$ can be shown as in \cite{baltagdynamic2018}: by Def.~\ref{def:diachronic}, the soundness of the axiom $\bigcirc f_a \leftrightarrow f_a \lor f^\tau_{N(a)}$ is shown in the same way.  

As an example, we prove the validity of $\square N_{ab} \leftrightarrow N_{ab} \lor sim^\omega_{ab}$. 

Consider a model ${M= \langle{\mathcal{N},\mathcal{V},\omega,\tau}\rangle}$. $M\models\square N_{ab}$ iff, by Def.~\ref{def:semantics}, $M_\square\models N_{ab}$ iff, by Def.~\ref{def:modelupdate} either $(a,b)\in\mathcal{N}$, or $\frac{|(\mathcal{V}(a)\cap\mathcal{V}(b))\cup(\mathcal{F}\setminus({\mathcal{V}(a)}\cup{{\mathcal{V}(b)}))|}}{|\mathcal{F}|}\geq{\omega}$. 

This holds iff $M\models{N_{ab}}$ or $\frac{|(\mathcal{V}(a)\cap\mathcal{V}(b))\cup(\mathcal{F}\setminus({\mathcal{V}(a)}\cup{{\mathcal{V}(b)}))|}}{|\mathcal{F}|}\geq{\omega}$. 

We now show that $\frac{|(\mathcal{V}(a)\cap\mathcal{V}(b))\cup(\mathcal{F}\setminus({\mathcal{V}(a)}\cup{{\mathcal{V}(b)}))|}}{|\mathcal{F}|}\geq{\omega}$ iff $M\models sim^\omega_{ab}$. 

[$\Rightarrow$] If $\frac{|(\mathcal{V}(a)\cap\mathcal{V}(b))\cup(\mathcal{F}\setminus({\mathcal{V}(a)}\cup{{\mathcal{V}(b)}))|}}{|\mathcal{F}|}\geq{\omega}$, there exist a subset $E\subseteq\mathcal{F}$, namely the set $(\mathcal{V}(a)\cap\mathcal{V}(b))\cup(\mathcal{F}\setminus({\mathcal{V}(a)}\cup{{\mathcal{V}(b)}}))$  such that for all $f\in{E}$, $M\models f_a\leftrightarrow f_b$ (indeed the set $(\mathcal{V}(a)\cap\mathcal{V}(b))\cup(\mathcal{F}\setminus({\mathcal{V}(a)}\cup{{\mathcal{V}(b)}))}$ contains by definition all and only those features that either both agents have or that both do not have). This in turn implies that $M\models\bigvee_{\{E\subseteq{\mathcal{F}}: \frac{|E|}{|\mathcal{F}|}\geq{\omega}\}}\bigwedge_{f\in{E}}(f_a\leftrightarrow{f_b})$. Thus, $M\models{sim^\omega_{ab}}$. 

[$\Leftarrow$] 
Now, assume that $M\models{sim^\omega_{ab}}$. This holds iff $M\models\bigvee_{\{E\subseteq{\mathcal{F}}: \frac{|E|}{|\mathcal{F}|}\geq{\omega}\}}\bigwedge_{f\in{E}}(f_a\leftrightarrow{f_b})$. This means that there exist a subset $E\subseteq\mathcal{F}$, such that $\frac{|E|}{|\mathcal{F}|}\geq\omega$, and such that for all $f\in E$, $M\models f_a\leftrightarrow f_b$, i.e. $f\in V(a)$ iff $f\in V(b)$. From this, it is clear that $E$ must be a subset of $(\mathcal{V}(a)\cap\mathcal{V}(b))\cup(\mathcal{F}\setminus({\mathcal{V}(a)}\cup{{\mathcal{V}(b)}))}$. From this and the fact that $\frac{|E|}{|\mathcal{F}|}\geq\omega$, we can conclude that $\frac{|(\mathcal{V}(a)\cap\mathcal{V}(b))\cup(\mathcal{F}\setminus({\mathcal{V}(a)}\cup{{\mathcal{V}(b)}))}|}{|\mathcal{F}|}\geq\omega$.

We thus proved that $\frac{|(\mathcal{V}(a)\cap\mathcal{V}(b))\cup(\mathcal{F}\setminus({\mathcal{V}(a)}\cup{{\mathcal{V}(b)}))|}}{|\mathcal{F}|}\geq{\omega}$ iff $M\models sim^\omega_{ab}$.
From above we know that $M\models \square N_{ab}$ iff $M\models f_a$ or $\frac{|(\mathcal{V}(a)\cap\mathcal{V}(b))\cup(\mathcal{F}\setminus({\mathcal{V}(a)}\cup{{\mathcal{V}(b)}))|}}{|\mathcal{F}|}\geq{\omega}$. We can therefore conclude that  $M\models \square N_{ab}$ iff $M\models f_a$ or $M\models sim^\omega_{ab}$. 

The soundness of the axiom $\bigcirc N_{ab} \leftrightarrow N_{ab} \lor sim^\omega_{ab}$ is proven in the same way.
As usual, the soundness of the distributivity of the dynamic operators over conjunction and the clauses for negation can be proven by induction on the length of formulas.
Finally, validity preservation of the inference rule in Table~\ref{table:axiom1} can be shown by induction on the structure of $\varphi$.

[\textit{Completeness}] Completeness is proven in the standard way by defining a translation from the dynamic language into the static fragment of the language, see for instance \cite{DEL, kooiredexprs}.

\end{proof}

\subsection*{Proof of Proposition~\ref{prop:equivalencemonotonic}}

\begin{proof}
    
        [First point]

        [$\Rightarrow$]  Let a model ${M= \langle{\mathcal{N},\mathcal{V},\omega,\tau}\rangle}$ be given. Assume, towards a contradiction, that $\triangle$ is equivalent to $\bigcirc$ on $M$, and that $M\not\models{\psi_{\triangle}}$. By the definition of $\psi_\triangle$ in Def.~\ref{def:psiabbreviations}, $M\not\models{\bigwedge_{a,b\in \mathcal{A}}({N_{ab}}\lor\neg{sim^\omega_{ab}})}$. Since $M\not\models{\bigwedge_{a,b\in \mathcal{A}}({N_{ab}}\lor\neg{sim^\omega_{ab}})}$, it follows that there are $a,b\in \mathcal{A}$ such that $M\models{\neg{N_{ab}}\wedge{sim^\omega_{ab}}}$. From this, and the reduction axioms for the dynamic modalities $\bigcirc$ and $\triangle$ in Table~\ref{table:axiom1}, respectively, it follows that $M\models{\bigcirc{N_{ab}}}$ and $M\not\models{\triangle{N_{ab}}}$. By Def.~\ref{def:equivalences}, this contradicts the initial assumption that $\triangle$ is equivalent to $\bigcirc$ on $M$. 
        
        [$\Leftarrow$] Let a model ${M= \langle{\mathcal{N},\mathcal{V},\omega,\tau}\rangle}$ be given. Assume, towards a contradiction, that $M\models\psi_{\triangle}$, and that it is not the case that $\triangle$ is equivalent to $\bigcirc$ on $M$. 
        From the fact that $\triangle$ is not equivalent to $\bigcirc$ on $M$, by Obs.~\ref{obs:first}, it follows that it must be the case that $M_{\triangle}$ and $M_\bigcirc$ differ on whether they satisfy some atomic formula. 
        By Obs.~\ref{observation:inclusion}, we know that: (i) for all $a\in \mathcal{A}$, for all $f\in{\mathcal{F}}$, $M\models{\bigcirc{f_a}}$ iff $M\models{\triangle{f_a}}$ (a synchronous update and a diffusion update modifies in the same way the features of the agents); (ii) for all $a,b\in \mathcal{A}$, if $M\models{\triangle{N_{ab}}}$, then $M\models{\bigcirc{N_{ab}}}$. From (i) and (ii), and the fact that $M_{\triangle}$ and $M_\bigcirc$ differ on whether they satisfy some atomic formula, it must be the case that there are $a,b,\in \mathcal{A}$ such that $M\models{\bigcirc{N_{ab}}}$ and $M\not\models{\triangle{N_{ab}}}$. By the reduction axioms for the modality $\triangle$ in Table~\ref{table:axiom1} and the fact that $M\not\models{\triangle{N_{ab}}}$, it follows that  $M\not\models{{N_{ab}}}$. By the facts that $M\not\models{{N_{ab}}}$ and that $M\models{\bigcirc{N_{ab}}}$, by the reduction axioms for the modality $\bigcirc$, it must be the case that $M\models{{sim^\omega_{ab}}}$. 
        We therefore know that it is both the case that  $M\not\models{{N_{ab}}}$ and that $M\models{{sim^\omega_{ab}}}$. Therefore for some $a,b\in\mathcal{A}$, it is true that $M\models{\neg{N_{ab}}\wedge{sim^\omega_{ab}}}$, i.e. $M\models\bigvee_{a,b\in\mathcal{A}}(\neg N_{ab}\wedge sim^\omega_{ab})$, which implies that $M\not\models\bigwedge_{a,b\in\mathcal A}(N_{ab}\lor \neg sim^\omega_{ab})$. This means that $M\not\models\psi_\triangle$. This contradicts the initial assumption that $M\models{\psi_\triangle}$.

        [Second point]

        [$\Rightarrow$] Let a model ${M= \langle{\mathcal{N},\mathcal{V},\omega,\tau}\rangle}$ be given. Assume, towards a contradiction, that, $\square$ is equivalent to $\bigcirc$ on $M$, and that 
$M\not\models\psi_\square$. From the definition of $\psi_\square$ in Def.~\ref{def:psiabbreviations}, it follows that $M\not\models{\bigwedge_{a\in \mathcal{A}}\bigwedge_{f\in{F}}({f_a}\lor\neg{f^\tau_{N(a)}})}$. Since, $M\not\models{\bigwedge_{a\in \mathcal{A}}\bigwedge_{f\in{F}}({f_a}\lor\neg{f^\tau_{N(a)}})}$, it follows that there are $a\in \mathcal{A}$, and $f\in{\mathcal{F}}$ such that $M\models{{\neg{f_a}}\wedge{f^\tau_{N(a)}}}$. From this and the reduction axioms for the dynamic modalities $\bigcirc$ and $\square$ in Table~\ref{table:axiom1}, it follows that $M\models{\bigcirc{f_a}}$ and $M\not\models{\square{f_a}}$. By Def.~\ref{def:equivalences}, this contradicts the initial assumption that $\square$ is equivalent to $\bigcirc$ on $\mathcal{M}$. 
        
        [$\Leftarrow$] Let a model ${M= \langle{\mathcal{N},\mathcal{V},\omega,\tau}\rangle}$ be given. Assume, towards a contradiction that, $M\models\psi_\square$, and that it is not the case that $\square$ and $\bigcirc$ are equivalent on $M$. From this and Def.~\ref{def:equivalences}, we know that $M_{\square}\neq M_\bigcirc$. By Obs.~\ref{obs:first}, it follows that it must be the case that $M_{\square}$ and $M_\bigcirc$ differ on whether they satisfy some atomic proposition. By Obs.~\ref{observation:inclusion}, we know that: (i) for all $a,b\in \mathcal{A}$, $M\models{\bigcirc{N_{ab}}}$ iff $M\models{\square{N_{ab}}}$ (a network update and a synchronous update modifies the network in exactly the same way); (ii) for all $a\in \mathcal{A}$, for all $f\in{\mathcal{F}}$, if $M\models{\square{f_{a}}}$, then $M\models{\bigcirc{f_{a}}}$. From (i) and (ii) and the fact that $M_{\square}$ and $M_\bigcirc$ differ on whether they satisfy some atomic proposition, it must the case that there are $a\in \mathcal{A}$ and $f\in{\mathcal{F}}$, such that $M\models{\bigcirc{f_{a}}}$ and $M\not\models{\square{f_{a}}}$. By the fact that $M\not\models{\square{f_{a}}}$ and the reduction axioms for the $\square$ modality in Table~\ref{table:axiom1}, it follows that $M\not\models{{f_{a}}}$. By the fact that $M\not\models{{f_{a}}}$ and $M\models{\bigcirc{f_{a}}}$, by the reduction axioms for the $\bigcirc$ modality in Table~\ref{table:axiom1}, it follows that it must be the case that $M\models{{f^\tau_{N(a)}}}$. We therefore know that it is the case that: $M\not\models{{f_{a}}}$ and $M\models{{f^\tau_{N(a)}}}$. Therefore, there exist $a\in\mathcal{A}$ and $f\in\mathcal{F}$ such that  $M\models \neg{f_a} \wedge f^\tau_{N(a)}$. From this, it follows that $M\models\bigvee_{a\in\mathcal{A}}\bigvee_{f\in\mathcal{F}}(\neg f_a \wedge f^\tau_{N(a)})$, and, therefore,  $M\not\models\bigwedge_{a\in\mathcal{A}}\bigwedge_{f\in\mathcal{F}}( f_a \lor \neg f^\tau_{N(a)})$. This means that $M\not\models\psi_\square$. We have therefore reached a contradiction with the initial assumption that $M\models\psi_\square$.

[Third point]

[$\Rightarrow$] Let a model ${M= \langle{\mathcal{N},\mathcal{V},\omega,\tau}\rangle}$ be given. Assume, towards a contradiction, that, $\triangle\square$ is equivalent to $\bigcirc$ on $M$, and that $M\not\models\psi_{\triangle\square}$. By the definition of $\psi_{\triangle\square}$ in Def.~\ref{def:psiabbreviations}, we know that $M\not\models\bigwedge_{{a,b}\in \mathcal{A}}(\neg{N_{ab}}\rightarrow({sim^{\omega}_{ab}}\leftrightarrow\triangle{{sim^\omega_{ab}})})$. 
From this, it follows that, for some $a,b\in \mathcal{A}$, $M\models{\neg{N_{ab}}}$ and $M\not\models{({sim^{\omega}_{ab}}\leftrightarrow\triangle{{sim^\omega_{ab}})}}$. One of the following two must be the case: (i) $M\models{{sim^{\omega}_{ab}}}$ and $M\not\models\triangle{{sim^\omega_{ab}}}$; (ii) $M\not\models{{sim^{\omega}_{ab}}}$ and $M\models\triangle{{sim^\omega_{ab}}}$. Assume that (i) is the case: the facts that $M\models{{sim^{\omega}_{ab}}}$ and $M\not\models\triangle{{sim^\omega_{ab}}}$, together with the fact that $M\models{\neg{N_{ab}}}$ and the reduction axioms in Table~\ref{table:axiom1}, imply that $M\models{\bigcirc{N_{ab}}}$ and $M\not\models{\triangle\square{N_{ab}}}$. This implies that $\bigcirc$ is not equivalent to $\triangle\square$ on $M$, contrary to our initial assumption. It must therefore be the case that (ii) holds. Assume that (ii) is true, i.e. that $M\not\models{{sim^{\omega}_{ab}}}$ and $M\models\triangle{{sim^\omega_{ab}}}$. These assumptions, together with the fact that $M\not\models{N_{ab}}$ and the reduction axioms in Table~\ref{table:axiom1}, imply that $M\not\models{\bigcirc{N_{ab}}}$ and $M\models{\triangle\square{N_{ab}}}$. By Def.~\ref{def:equivalences}, this contradicts the initial assumption that $\bigcirc$ and $\triangle\square$ are equivalent on $M$.
Since neither (i) nor (ii) are possible, we have established that if $\bigcirc$ is equivalent $\triangle\square$ on $M$, then $M\models\psi_{\triangle\square}$.
        
        [$\Leftarrow$] Let a model ${M= \langle{\mathcal{N},\mathcal{V},\omega,\tau}\rangle}$ be given. Assume, towards a contradiction that $M\models\psi_{\triangle\square}$, and that $\triangle\square$ is not equivalent to $\bigcirc$ on $M$. From the fact that $\triangle\square$ is not equivalent to $\bigcirc$ on $M$, by Def.~\ref{def:equivalences}, it follows that $M_{\triangle\square}\neq M_\bigcirc$. By Obs.~\ref{obs:first}, we know that $M_{\triangle\square}$ and $M_\bigcirc$ must differ on whether they satisfy some atomic proposition. By Obs.~\ref{observation:inclusion}, and by Def.~\ref{def:modelupdate} and Def.~\ref{def:diachronic}, we know that for all $a\in \mathcal{A}$, and all $f\in{\mathcal{F}}$, $M\models{\bigcirc{f_{a}}}$ iff $M\models{\triangle{f_{a}}}$ iff $M\models{\triangle\square{f_{a}}}$ (informally, this simply mean that, since a network update does not affect the agent's features, and a synchronous update and a diffusion update change the features in the same way, the features of the agents after one synchronous update are the same as those after one diffusion update followed by a subsequent network update). From this and the fact that $M_{\triangle\square}$ and $M_\bigcirc$ must differ on whether they satisfy some atomic proposition, it follows that one of the following two cases must hold: (i) there are $a,b$ such that $M\models{\triangle\square{N_{ab}}}$, and $M\not\models{\bigcirc{N_{ab}}}$; (ii) there are $a,b$ such that $M\not\models{\triangle\square{N_{ab}}}$, and $M\models{\bigcirc{N_{ab}}}$. 
        
        Assume that (i) is the case, i.e. there are $a,b$ such that $M\models{\triangle\square{N_{ab}}}$, and $M\not\models{\bigcirc{N_{ab}}}$; 
        from the fact that $M\not\models{\bigcirc{N_{ab}}}$ and the reduction axioms for $\bigcirc$ in Table~\ref{table:axiom1}, we know that $M\not\models{N_{ab}\lor sim^\omega_{ab}}$, which implies that $M\models\neg{N_{ab}}$ and $M\models\neg sim^\omega_{ab}$; furthermore, from the fact that $M\models\triangle\square N_{ab}$, and that $M\not\models N_{ab}$, we know that $M\models\triangle sim^\omega_{ab}$. If we put these together, we know that $M\models\neg{N_{ab}}$, $M\models \neg sim^\omega_{ab}$, and $M\models\triangle sim^\omega_{ab}$ at the same time: this implies that $M\models\bigvee_{a,b\in\mathcal{A}} (\neg N_{ab} \wedge \neg sim^\omega_{ab} \wedge \triangle sim^\omega_{ab})$. From this, it follows that $M\not\models\bigwedge_{a,b\in\mathcal{A}} (\neg N_{ab} \rightarrow ( sim^\omega_{ab} \leftrightarrow \triangle sim^\omega_{ab}))$: by Def. ~\ref{def:psiabbreviations}, it follows that $M\not\models\psi_{\triangle\square}$, which contradicts our initial assumption that $M\models\psi_{\triangle\square}$. 
        
        Since (i) is not possible, it must be the case that (ii) holds, i.e. there are $a,b$ such that $M\not\models{\triangle\square{N_{ab}}}$, and $M\models{\bigcirc{N_{ab}}}$. From the fact that $M\not\models{\triangle\square N_{ab}}$ and the reduction axioms for $\triangle$ and $\square$, we know that $M\not\models\triangle sim^\omega_{ab}$ and $M\not\models N_{ab}$. From the fact that $M\models\bigcirc N_{ab}$ and $M\not\models N_{ab}$, by the reduction axioms for $\bigcirc$, we know that $M\models sim^\omega_{ab}$. Summarising the facts above, we therefore know that $M\models \neg N_{ab}$, $M\models sim^\omega_{ab}$ and $M\models \neg \triangle sim^\omega_{ab}$: this means that $M\models\bigvee_{a,b\in\mathcal{A}} (\neg N_{ab} \wedge sim^\omega_{ab} \wedge \neg \triangle sim^\omega_{ab})$. This implies the fact that $M\not\models\bigwedge_{{a,b}\in \mathcal{A}}(\neg{N_{ab}}\rightarrow({sim^{\omega}_{ab}}\leftrightarrow\triangle{{sim^\omega_{ab}})})$. By Def. ~\ref{def:psiabbreviations}, it follows that $M\not\models\psi_{\triangle\square}$, which contradicts our initial assumption that $M\models\psi_{\triangle\square}$. 

        Since neither (i) nor (ii) are possible, we have established that if $M\models\psi_{\triangle\square}$, then $\bigcirc$ must be equivalent to $\triangle\square$ on $M$.
        
        
[Fourth point]

[$\Rightarrow$] Let a model ${M= \langle{\mathcal{N},\mathcal{V},\omega,\tau}\rangle}$ be given. Assume, towards a contradiction, that, for some arbitrary $n>0$, $\square\triangle^n$ is equivalent to $\bigcirc$ on $\mathcal{M}$ and that $M\not\models\psi_{\square\triangle^n}$. From the fact that $M\not\models\psi_{\square\triangle^n}$, by Def.~\ref{def:psiabbreviations}, it follows that $M\not\models\bigwedge_{a\in \mathcal{A}}\bigwedge_{f\in{F}}(\neg{f_a}\rightarrow({f^\tau_{N(a)}}\leftrightarrow\bigvee_{0\leq{i}\leq{n-1}}{\square\triangle^{i}{f^\tau_{N(a)}})})$. From this, it follows that there exist $a\in \mathcal{A}$ and $f\in{\mathcal{F}}$ such that $M\models{\neg{f_{a}}}$ and $M\not\models({f^\tau_{N(a)}}\leftrightarrow\bigvee_{0\leq{i}\leq{n-1}}{\square\triangle^{i}{f^\tau_{N(a)}})}$. 
One of the following two must hold: (i) $M\models{f^\tau_{N(a)}}$ and  $M\not\models\bigvee_{0\leq{i}\leq{n-1}}{\square\triangle^{i}{f^\tau_{N(a)}}}$, or (ii) $M\not\models{f^\tau_{N(a)}}$ and  $M\models\bigvee_{0\leq{i}\leq{n-1}}{\square\triangle^{i}{f^\tau_{N(a)}}}$.

Assume that (i) is the case and thus that $M\models{f^\tau_{N(a)}}$ and  $M\not\models\bigvee_{0\leq{i}\leq{n-1}}{\square\triangle^{i}{f^\tau_{N(a)}}}$. Since $M\models{f^\tau_{N(a)}}$, by the reduction axiom for $\bigcirc$, we know that $M\models\bigcirc f_a$. 
At the same time, since we know that $M\not\models\bigvee_{0\leq{i}\leq{n-1}}{\square\triangle^{i}{f^\tau_{N(a)}}}$, we know that for all $0\leq{i}\leq{n-1}$ $M\not\models{\square\triangle^{i}{f^\tau_{N(a)}}}$: this simply means that after a network update, there is no sequence of diffusion update after which the agent $a$ has social conformity pressure to adopt feature $f$. From this and the fact that $M\not\models f_a$, by the reduction axioms of for the dynamic modalities, we know that $M\not\models \square\triangle f_a$. Therefore, it is both the case that $M\models\bigcirc f_a$, and $M\not\models \square\triangle^n f_a$, which, by Def.~\ref{def:equivalences}, contradicts the initial assumption that $\bigcirc$ is equivalent to $\square\triangle^n$ on $M$.

Since (i) is not possible, it must be the case that (ii) holds, i.e. it is the case that $M\not\models{f^\tau_{N(a)}}$ and  $M\models\bigvee_{0\leq{i}\leq{n-1}}{\square\triangle^{i}{f^\tau_{N(a)}}}$. From the fact that $M\not\models f^\tau_{N(a)}$, and the fact that $M\not\models f_a$, by the reduction axioms for $\bigcirc$, it follows that $M\not\models\bigcirc f_a$. From the fact that $M\not\models{f^\tau_{N(a)}}$ and  $M\models\bigvee_{0\leq{i}\leq{n-1}}{\square\triangle^{i}{f^\tau_{N(a)}}}$, it follows that there is an $i\leq(n-1)$ such that $M\models \square\triangle^i f^\tau_{N(a)}$ (this means that, after a network update, at some point of a sequence of further diffusion update, agent $a$ has pressure to adopt feature $f$). From this and the reduction axioms for the dynamic modality $\triangle$, it follows that $M\models \square\triangle^{i+1} f_a$. Since $i\leq{n-1}$, and by the fact that features cannot be abandoned, it follows that $M\models{\square\triangle^n f_a}$. Thus, it is both true that $M\not\models \bigcirc f_a$, and $M\models \square\triangle^n f_a$. By Def.~\ref{def:equivalences}, this contradicts the initial assumption that $\square\triangle^n$ is equivalent to $\bigcirc$ on $M$.

Since neither (i) nor (ii) are possible, we have established that if $\bigcirc$ is equivalent $\square\triangle^n$ on $M$, then $M\models\psi_{\square\triangle^n}$.

    [$\Leftarrow$] 
Let a model ${M= \langle{\mathcal{N},\mathcal{V},\omega,\tau}\rangle}$ be given. Assume, towards a contradiction, that, for some $n>0$,  $M\models\psi_{\square\triangle^n}$ and that $\bigcirc$ is not equivalent to $\square\triangle^n$ in $M$.
From the fact that $\bigcirc$ is not equivalent to $\square\triangle^n$ in $M$, by Def.~\ref{def:equivalences}, it follows that $M_{\square\triangle^n}\neq M_\bigcirc$. Thus, by Obs.~\ref{obs:first}, we know that $M_{\square\triangle^n}$ and $ M_\bigcirc$ must differ on whether they satisfy some atomic proposition.
By Obs.~\ref{observation:inclusion} and by Def.~\ref{def:modelupdate} and Def.~\ref{def:diachronic}, we know that, for all $a,b\in \mathcal{A}$, $M\models{\bigcirc{N_{ab}}}$ iff $M\models{\square{N_{ab}}}$ iff $M\models{\square\triangle^n{N_{ab}}}$. Informally, this follows from the fact that, since a diffusion update does not affect the agent's features, and a synchronous update and a network update change the network structure in the same way, the connections between the agents after one synchronous update are the same as those obtained after one network update followed by multiple subsequent diffusion update.
From this and the fact that $M_{\square\triangle^n}$ and $ M_\bigcirc$ must differ on whether they satisfy some atomic proposition, there must exist $a\in \mathcal{A}$ and $f\in{\mathcal{F}}$, s.t. it is not the case that $M\models{\square\triangle^{n}f_a}$ iff $M\models{\bigcirc{f_a}}$. Therefore, either one of the following cases must hold: (i) $M\models{\square\triangle^{n}f_a}$ and $M\not\models{\bigcirc{f_a}}$, or (ii) $M\not\models{\square\triangle^{n}f_a}$ and $M\models{\bigcirc{f_a}}$.

Assume that (i):  $M\models{\square\triangle^{n}{f_a}}$ and $M\not\models{\bigcirc{f_a}}$. By the fact $M\not\models{\bigcirc{f_a}}$ and the reduction axioms for $\bigcirc$, we know that $M\not\models{f^\tau_{N(a)}}$ and $M\not\models f_a$. From the fact that $M\models{\square\triangle^{n}{f_a}}$, we know that there exist an $0\leq{i}\leq{n-1}$ such that $M\models\square\triangle^{i}f^\tau_{N(a)}$: this simply means that, at some point, after a network-update and potentially after subsequent diffusion updates, $a$ has pressure to adopt $f$. From this, it follows that $M\models\bigvee_{0\leq{i}\leq{n-1}}{\square\triangle^{i}{f^\tau_{N(a)}})}$. Therefore, from the above we know that $M\models\neg{f_a}$, $M\models\neg{f^\tau_{N(a)}}$ and $M\models\bigvee_{0\leq{i}\leq{n-1}}{\square\triangle^{i}{f^\tau_{N(a)}})}$. This imply that $M\models\bigvee_{a\in\mathcal{A}}\bigvee_{f\in\mathcal{F}}(\neg{f_a} \wedge \neg{f^\tau_{N(a)}} \wedge \bigvee_{0\leq{i}\leq{n-1}}{\square\triangle^{i}{f^\tau_{N(a)}})})$. Therefore $M\not\models\bigwedge_{a\in \mathcal{A}}\bigwedge_{f\in{F}}(\neg{f_a}\rightarrow({f^\tau_{N(a)}}\leftrightarrow\bigvee_{0\leq{i}\leq{n-1}}{\square\triangle^{i}{f^\tau_{N(a)}})})$. Thus, by Def.~\ref{def:psiabbreviations}, we know that $M\not\models\psi_{\square\triangle^n}$, which contradicts our initial assumption that $M\models\psi_{\square\triangle^n}$.
 

Since (i) cannot be the case, it must be the case that (ii): $M\not\models{\square\triangle^{n}{f_a}}$ and $M\models{\bigcirc{f_a}}$. 

From the fact $M\not\models{\square\triangle^{n}{f_a}}$, 
by the reduction axioms for the dynamic modalities, we know two things: $M\not\models{f_a}$, and there does not exist $0\leq i \leq n-1$, such that $M\models \square\triangle^i f^\tau_{N(a)}$ (this simply means that at no point after a network update and subsequent diffusion updates $a$ has pressure to adopt $f$; indeed, if this was the case then $a$ would at some point adopt $f$, and will never abandon it).
From the fact that $M\not\models f_a$ and that $M\models\bigcirc f_a$, by the reduction axiom for $\bigcirc$, we know that $M\models f^\tau_{N(a)}$. 
Summarising the above we know that: $M\models\neg f_a$, $M\models f^\tau_{N(a)}$ and $M\models \neg\bigvee_{0\leq{i}\leq{n-1}}{\square\triangle^{i}{f^\tau_{N(a)}}}$.
Thus, $M\models \bigvee_{a\in\mathcal{A}}\bigvee_{f\in\mathcal{F}}(\neg f_a \wedge f^\tau_{N(a)} \wedge \neg\bigvee_{0\leq{i}\leq{n-1}}{\square\triangle^{i}{f^\tau_{N(a)}})}$. Thus, $M\not\models\bigwedge_{a\in \mathcal{A}}\bigwedge_{f\in{F}}(\neg{f_a}\rightarrow({f^\tau_{N(a)}}\leftrightarrow\bigvee_{0\leq{i}\leq{n-1}}{\square\triangle^{i}{f^\tau_{N(a)}})})$. By Def.~\ref{def:psiabbreviations}, it follows that $M\not\models\psi_{\square\triangle^n}$, contrary to the initial assumption that $M\models\psi_{\square\triangle^n}$.

Since neither (i) nor (ii) are possible, we have established that if $M\models\psi_{\square\triangle^n}$, then $\bigcirc$ is equivalent to $\square\triangle^n$ on $M$.



\end{proof}

\end{document}